\documentclass[aps,prresearch,twocolumn,superscriptaddress,longbibliography,nofootinbib,floatfix]{revtex4-2}

\usepackage{amsmath,amssymb,amsthm,mathtools}
\usepackage{graphicx}
\usepackage{tikz}
\usetikzlibrary{positioning,arrows.meta,calc}
\usepackage[colorlinks=true,linkcolor=blue,citecolor=blue,urlcolor=blue]{hyperref}

\newtheorem{theorem}{Theorem}
\newtheorem{lemma}{Lemma}
\newtheorem{corollary}{Corollary}
\newtheorem{proposition}{Proposition}
\theoremstyle{remark}
\newtheorem{remark}{Remark}

\newcommand{\id}{\openone}
\newcommand{\Tr}{\mathrm{Tr}}
\newcommand{\Sym}{\mathrm{Sym}}

\DeclareMathOperator*{\eigmin}{\lambda_{\min}}
\DeclareMathOperator*{\eigmax}{\lambda_{\max}}

\begin{document}

\title{Analytic properties of cross-click operators in passive multi-basis photodetection:\\
monotonicity, exact convergence rates, and dimension reduction for quantum key distribution}

\author{Zhangdong Ye}
\email{yezhangdong@126.com}
\affiliation{Independent researcher, Ningbo 315100, China}

\date{July 5, 2026}

\begin{abstract}
In security analyses of quantum key distribution (QKD) and entanglement verification with realistic threshold detectors, \emph{cross-click operators}---the positive-operator-valued-measure (POVM) elements describing simultaneous clicks of detectors belonging to different measurement bases---play a central role: the growth of their minimum eigenvalue $f^{(n)}$ on the $n$-photon subspace is used to bound the weight of multiphoton contributions and thereby to reduce an infinite-dimensional optics problem to a low-dimensional one. In the works that introduced this technique, the required growth---monotonicity of $f^{(n)}$ in $n$ and convergence $f^{(n)}\to 1$---was \emph{verified numerically} on finitely many photon-number sectors, and while analytic subspace-population bounds have since been obtained for specific protocols and, recently, for generic passive setups within the flag-state squashing framework, a sharp analytic characterization of its spectral behavior has been lacking. Here we provide one. We show that for an arbitrary passive linear-optical analyzer with threshold detectors of arbitrary (possibly mismatched) efficiencies and dark-count probabilities, every ``silence'' operator is the second quantization $\Gamma(A)$ of an explicit single-photon contraction $A$, so that its restriction to the $n$-photon sector is the symmetric tensor power $A^{\otimes n}|_{\Sym^n}$. From this structure we prove: (i) $f^{(n+1)}\ge f^{(n)}$ for all $n$, via a simple operator inequality; (ii) two-sided exponential bounds $\max_b \gamma_b\lVert A_b\rVert^n \le 1-f^{(n)} \le \sum_b \gamma_b\lVert A_b\rVert^n$, which identify the \emph{exact} asymptotic convergence rate through an explicit single-photon spectral quantity, in closed form for the six-state polarization analyzer; (iii) for ideal detectors and $n\ge1$, the exact value $f^{(n)}=1-\sum_b p_b^{\,n}$; and (iv) an exact factorization $1-f^{(n_A,n_B)}=(1-f_A^{(n_A)})(1-f_B^{(n_B)})$ for the joint two-party cross-click operator, which reduces all bipartite properties to single-party ones. The proofs are dimension-agnostic and apply verbatim to polarization, time-bin, and spatial-mode analyzers with any number of bases. As an application we derive closed-form photon-number weight bounds of the type used in QKD security proofs with detection-efficiency mismatch, replacing large-scale Fock-space numerics by explicit formulas and making the associated security claims rigorous for \emph{all} photon numbers.
\end{abstract}

\maketitle

\section{Introduction}\label{sec:intro}

Security proofs of quantum key distribution (QKD) \cite{BB84,Renner2008,Scarani2009,Xu2020} must ultimately refer to the devices actually used in an experiment. On the detection side, practical receivers consist of passive linear optics followed by threshold single-photon detectors, whose efficiencies are neither unity nor identical. Detection-efficiency mismatch is not a technicality: it opens well-documented attacks \cite{Makarov2006,Lydersen2010,Sajeed2015}, and incorporating it into security proofs for realistic devices \cite{Lutkenhaus2000,GLLP2004} has been a long-standing effort \cite{Fung2009,Bochkov2019,Trushechkin2020,Zhang2021}; on the source side, the analogous imperfections are routinely handled by decoy-state methods \cite{Hwang2003,LoMaChen2005,WangXB2005}.

A structural difficulty is that optical receivers act on an infinite-dimensional Fock space, while the information-theoretic machinery of security proofs---in particular numerical approaches based on convex optimization \cite{Coles2016,Winick2018} and dimension-reduction techniques \cite{Upadhyaya2021}---requires finite-dimensional models. Squashing models \cite{Beaudry2008,Tsurumaru2008,Gittsovich2014} provide one route, but they do not exist for all receivers and typically fail in the presence of efficiency mismatch. An alternative route, introduced by Zhang and L\"utkenhaus in the context of entanglement verification \cite{Zhang2017} and subsequently used in QKD security proofs with mismatched detectors \cite{Zhang2021} and in normalized security-analysis frameworks for four--six-state protocols \cite{Tannous2019,YeThesis} and, in a related numerics-driven form, for quantum secure direct communication \cite{YeQSDC2021}, exploits \emph{cross clicks}: coincidence events in which detectors belonging to \emph{different} measurement bases fire simultaneously.

The logic is simple and powerful. Let $F_C$ denote the POVM element of a cross click and let
\begin{equation}
  f^{(n)} \;=\; \eigmin\!\big(F_C^{(n)}\big)
\end{equation}
be its minimum eigenvalue on the $n$-photon subspace. If $f^{(n)}$ grows with $n$ and approaches $1$, then a small observed cross-click probability $q_C$ certifies that the incoming light is predominantly supported on low photon numbers,
\begin{equation}\label{eq:weightlogic}
  q_C \;\ge\; \sum_n P_n f^{(n)} \;\ge\; f^{(N+1)} \sum_{n>N} P_n ,
\end{equation}
so that $\sum_{n\le N} P_n \ge 1 - q_C/f^{(N+1)}$. This converts a measurable coincidence rate into a rigorous bound on the weight of the low-photon-number block, which is exactly what dimension reduction requires.

The second inequality in Eq.~\eqref{eq:weightlogic} discards the sectors $n\le N$ using $f^{(n)}\ge0$ and then applies $f^{(n)}\ge f^{(N+1)}$ for $n>N$; its validity thus rests on the monotonicity (together with the nonnegativity) of $f^{(n)}$, while the convergence $f^{(n)}\to1$ is what makes the resulting bound strong at large cutoffs. In Refs.~\cite{Zhang2017,Zhang2021} and in the four--six-state framework of Ref.~\cite{YeThesis} these behaviors were established numerically, sector by sector---via semidefinite programming up to $20$ photons in Ref.~\cite{Zhang2017}, and by explicit diagonalization, with an accelerated construction of the mode-operator products, up to $60$ photons in Ref.~\cite{YeThesis}; indeed, Ref.~\cite{Zhang2017} explicitly posed an analytic proof of these monotonic behaviors as an open problem. Such computations are costly---exponentially so in the naive labeled-mode expansion---and they can never certify the property for \emph{all} $n$: a security statement that quantifies over all photon numbers cannot rest on a finite numerical check.

Analytic progress on closely related quantities has been made within the flag-state squashing framework. Li and L\"utkenhaus derived analytic subspace-population bounds from multi-click statistics for the unbalanced phase-encoded BB84 receiver \cite{LiLutkenhaus2020}, and Kamin and L\"utkenhaus recently obtained provably secure analytic \emph{lower bounds} on subspace populations for arbitrary passive linear-optical detection setups \cite{Kamin2024}; passive-detection imperfections also feature in recent phase-error-estimation analyses \cite{Tupkary2024,Wang2025}. These results bound the probability of suitably chosen detectable events from below on high-photon-number sectors, which suffices for flag-state squashing in prepare-and-measure protocols. What has been missing is an analytic characterization of the spectrum of the cross-click operator itself: a proof that its minimum eigenvalue---the quantity constrained in the entanglement-verification and joint two-party analyses of Refs.~\cite{Zhang2017,YeThesis}---is monotone for \emph{all} $n$, a determination of its exact asymptotic behavior with matching upper bounds, and a treatment of the genuinely bipartite joint cross-click operator. The physical intuition---many photons entering a passive multi-basis splitter will almost surely trigger detectors in more than one basis---strongly suggests that such a statement should exist.

The question addressed here is, in this sense, not a mathematical abstraction but one that arose from experimental and numerical studies of passive multi-basis receivers. In the six-state--four-state reference-frame-independent \cite{Laing2010} demonstration of Ref.~\cite{Tannous2019}, the detector efficiencies of the passive polarization analyzers were characterized as part of the analysis, making efficiency mismatch a concrete, measured device parameter rather than an idealization. It was in the author's subsequent doctoral work, completed in 2020 and reported in Ref.~\cite{YeThesis}, that the cross-click operators of precisely these receivers became the objects of the sector-by-sector numerical study described above; the resulting weight bounds, however, still rested on finite photon-number checks. The present work provides the missing analytic closure, and does so for a considerably more general class of passive analyzers.

Our starting point is the observation that every ``silence'' operator of a set of threshold detectors behind a passive network is the \emph{second quantization} $\Gamma(A)$ of an explicit contraction $A$ on the single-photon space, so that its $n$-photon block is a symmetric tensor power. All properties of the cross-click operator then follow from finite-dimensional matrix inequalities. Concretely, we prove the following, for an arbitrary passive analyzer with $m\ge 2$ detector groups (``bases''), arbitrary detector efficiencies $\eta_d\in(0,1]$, dark-count probabilities, internal losses, and an input space of arbitrary mode dimension:

\begin{enumerate}
\item[(R1)] \emph{Monotonicity} (Theorem~\ref{thm:mono}): $f^{(n+1)}\ge f^{(n)}$ for all $n\ge 0$.
\item[(R2)] \emph{Two-sided exponential bounds and the exact rate} (Theorem~\ref{thm:rate}): $\max_b\gamma_b\lVert A_b\rVert^n\le 1-f^{(n)}\le\sum_b\gamma_b\lVert A_b\rVert^n$, with $\lVert A_b\rVert$ an explicit single-photon spectral quantity, in closed form for the six-state analyzer; hence $\lim_n (1-f^{(n)})^{1/n}=\max_b\lVert A_b\rVert<1$ under a simple ``complementary coverage'' condition.
\item[(R3)] \emph{Exact values for ideal detectors} (Corollary~\ref{cor:ideal}): for $n\ge1$, $f^{(n)}=1-\sum_b p_b^{\,n}$, e.g. $f^{(n)}=1-3^{1-n}$ for the balanced passive six-state analyzer.
\item[(R4)] \emph{Exact factorization of the joint operator} (Theorem~\ref{thm:joint}): $1-f^{(n_A,n_B)}=(1-f_A^{(n_A)})(1-f_B^{(n_B)})$, which proves the bipartite monotonicity used in Refs.~\cite{Zhang2017,YeThesis} and explains, e.g., the observed degeneracy of the $n_B=0$ and $n_B=1$ curves.
\end{enumerate}

Beyond closing the logical gap, these results have practical consequences. The weight bounds entering security proofs become explicit finite-dimensional expressions---fully closed form for the six-state analyzer---eliminating both the computational cost and the truncation-error analysis of large Fock-space numerics; the exact convergence rate quantifies how the certification power of cross clicks degrades with efficiency mismatch; and the factorization (R4) collapses the two-party problem, which was previously handled by two-dimensional numerical scans, to two one-party curves.

The paper is organized as follows. Section~\ref{sec:setup} defines the detection model and the cross-click operator and establishes the second-quantization structure (Lemma~\ref{lem:gamma}). Section~\ref{sec:main} contains the main theorems with complete proofs. Section~\ref{sec:application} develops the application to photon-number weight bounds and dimension reduction in QKD. Section~\ref{sec:numerics} reports numerical illustrations and an independent Fock-space cross-check, together with a polynomial-time algorithm for exact sector-wise diagonalization. Section~\ref{sec:discussion} discusses extensions and open problems. Technical details are collected in the appendixes.

\section{Detection model and cross-click operators}\label{sec:setup}

\subsection{Passive multi-basis analyzers with threshold detectors}

\begin{figure}[tb]
\centering
\begin{tikzpicture}[scale=0.82, every node/.style={font=\scriptsize}]
  \draw[-{Latex[length=2mm]},thick] (-1.6,0) -- (-0.05,0) node[midway,above]{$\rho$};
  \draw[thick,fill=blue!8] (0,-0.45) rectangle (0.8,0.45);
  \node at (0.4,0) {\shortstack{$1{:}3$\\BS}};
  \foreach \y/\lab/\ba in {1.5/{$\mathbb{Z}$ basis (PBS $0^\circ$)}/Z, 0/{$\mathbb{X}$ basis (PBS $45^\circ$)}/X, -1.5/{$\mathbb{Y}$ basis (QWP+PBS)}/Y} {
    \draw[thick] (0.8,0.33*\y/1.5) .. controls (1.3,0.33*\y/1.5) and (1.5,\y) .. (2.0,\y);
    \draw[thick,fill=green!8] (2.0,\y-0.32) rectangle (3.6,\y+0.32);
    \node at (2.8,\y) {\lab};
  }
  \foreach \y/\da/\db in {1.5/{H}/{V}, 0/{D}/{A}, -1.5/{R}/{L}} {
    \draw[thick] (3.6,\y+0.16) -- (4.1,\y+0.35);
    \draw[thick] (3.6,\y-0.16) -- (4.1,\y-0.35);
    \draw[thick,fill=red!12] (4.1,\y+0.35) arc (90:-90:0.16) -- cycle;
    \draw[thick,fill=red!12] (4.1,\y-0.05) arc (90:-90:0.16) -- cycle;
    \node[right] at (4.35,\y+0.19) {$\eta_{\da},\varepsilon_{\da}$};
    \node[right] at (4.35,\y-0.51+0.3) {$\eta_{\db},\varepsilon_{\db}$};
  }
  \node[align=left] at (2.6,-2.6) {\emph{cross click}: detectors in $\ge 2$ different arms fire};
\end{tikzpicture}
\caption{Passive six-state polarization analyzer used as the running example. A polarization-independent $1{:}3$ splitter routes the incoming light to three polarization bases $\mathbb{Z}=\{H,V\}$, $\mathbb{X}=\{D,A\}$, $\mathbb{Y}=\{R,L\}$, each read out by two threshold detectors with efficiencies $\eta_d$ and dark-count probabilities $\varepsilon_d$. A \emph{cross click} is a coincidence between detectors belonging to different bases. All results in this paper hold for arbitrary passive analyzers of this type, with any number of input modes, detector groups, and detectors.}
\label{fig:setup}
\end{figure}
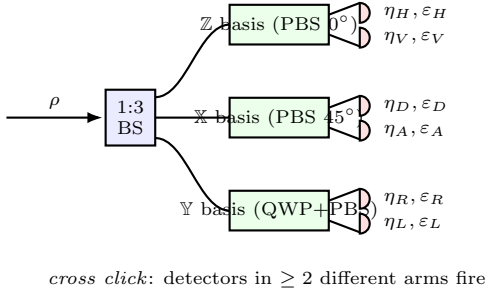

We consider the following general receiver, of which the passive six-state \cite{Bruss1998,BechmannPasquinucci1999} polarization analyzer of Fig.~\ref{fig:setup} is a special case (its six-state form is the device analyzed numerically in Ref.~\cite{YeThesis}; Ref.~\cite{Zhang2017} treated the analogous passive two-basis, four-detector analyzer). Light in $M$ optical modes, with single-photon Hilbert space $\mathfrak{h}\cong\mathbb{C}^M$, enters a passive (particle-number-conserving up to loss) linear-optical network terminating in $D$ threshold detectors labeled $d=1,\dots,D$. Detector $d$ has quantum efficiency $\eta_d\in(0,1]$ and dark-count probability $\varepsilon_d\in[0,1)$; as usual, the efficiency is modeled by a beam splitter of transmittance $\eta_d$ in front of an ideal threshold detector, and dark counts as independent classical clicks. Internal network losses are permitted and are absorbed into fictitious undetected loss modes, so that the network is described by an isometry $\widetilde{V}:\mathfrak{h}\to\mathfrak{h}_{\rm out}$ on the single-photon level, where $\mathfrak{h}_{\rm out}$ contains one mode per detector plus loss modes. We write $v_d^\dagger\in\mathfrak{h}^\ast$ for the row of $\widetilde V$ associated with detector mode $d$, i.e., $v_d\in\mathfrak{h}$ is the (sub-normalized) single-photon amplitude vector routed to detector $d$; isometry implies
\begin{equation}\label{eq:contraction}
  \sum_{d=1}^{D} v_d^{\phantom{\dagger}} v_d^\dagger \;\preceq\; \id_\mathfrak{h},
\end{equation}
with equality iff the network is lossless. For the balanced six-state analyzer, $M=2$, $D=6$, and $v_d=\sqrt{p_b}\,|e_d\rangle$ with splitting probability $p_b=\tfrac13$ and $|e_d\rangle\in\{|H\rangle,|V\rangle,|D\rangle,|A\rangle,|R\rangle,|L\rangle\}$.

The detectors are partitioned into $m\ge 2$ disjoint groups (``bases'')
\begin{equation}
  \{1,\dots,D\}=\mathbb{B}_1\sqcup\cdots\sqcup\mathbb{B}_m .
\end{equation}
For the six-state analyzer, $m=3$ and $\mathbb{B}_{\mathbb Z}=\{H,V\}$, $\mathbb{B}_{\mathbb X}=\{D,A\}$, $\mathbb{B}_{\mathbb Y}=\{R,L\}$.

Each run of the device produces a click pattern $S\subseteq\{1,\dots,D\}$, the set of detectors that fired (through photodetection or dark counts). Because threshold detectors respond only to the photon number arriving at them, all click-pattern POVM elements commute with the total input photon number $\hat n$; they are block diagonal in the photon-number decomposition $\mathcal{F}(\mathfrak{h})=\bigoplus_{n\ge0}\Sym^n(\mathfrak{h})$ of the input Fock space \cite{Zhang2017,YeThesis}. For an operator $X$ on $\mathcal{F}(\mathfrak{h})$ that is block diagonal in this sense we write $X^{(n)}$ for its block on the $n$-photon sector $\Sym^n(\mathfrak{h})$.

\subsection{Silence operators are second quantizations}

The elementary building blocks of our analysis are \emph{silence operators}: for a subset $S\subseteq\{1,\dots,D\}$, let $N_S$ denote the POVM element of the event ``no detector in $S$ clicks'' (detectors outside $S$ unconstrained). Recall that for a contraction $A$ on $\mathfrak h$ ($0\preceq A^\dagger A\preceq \id$; all our $A$'s will in fact be positive semidefinite), the \emph{second quantization} $\Gamma(A)$ is the operator on Fock space acting on the $n$-photon sector as the restricted tensor power \cite{BratteliRobinson},
\begin{equation}\label{eq:gammadef}
  \Gamma(A)\big|_{\Sym^n(\mathfrak h)} \;=\; A^{\otimes n}\big|_{\Sym^n(\mathfrak h)} .
\end{equation}
Physically, $\Gamma(A)$ ``applies $A$ to every photon independently.''

\begin{lemma}[Silence operators]\label{lem:gamma}
For any $S\subseteq\{1,\dots,D\}$,
\begin{equation}\label{eq:NS}
  N_S \;=\; \gamma_S\,\Gamma(A_S), \qquad
  \gamma_S=\!\prod_{d\in S}(1-\varepsilon_d), 
\end{equation}
\begin{equation}\label{eq:AS}
  A_S \;=\; \id_\mathfrak{h}-\sum_{d\in S}\eta_d\, v_d^{\phantom\dagger} v_d^\dagger,
  \qquad 0\preceq A_S\preceq \id_\mathfrak{h}.
\end{equation}
In particular, $N_S^{(n)}=\gamma_S\, A_S^{\otimes n}|_{\Sym^n(\mathfrak h)}$.
\end{lemma}

\begin{proof}
The no-click operator of a single threshold detector of efficiency $\eta$ monitoring output mode $d$ is $(1-\eta)^{\hat n_d}$: each arriving photon evades detection independently with probability $1-\eta$. Dark counts contribute the state-independent factor $(1-\varepsilon_d)$. Hence, on the output Fock space,
\begin{equation}
  N_S^{\rm out}=\gamma_S \prod_{d\in S}(1-\eta_d)^{\hat n_d}
  = \gamma_S\,\Gamma_{\rm out}(D_S),
\end{equation}
where $D_S$ is diagonal on $\mathfrak h_{\rm out}$ with entries $1-\eta_d$ for $d\in S$ and $1$ on all other detector and loss modes; the second equality is the standard identity $\prod_d x_d^{\hat n_d}=\Gamma(\bigoplus_d x_d)$, an immediate consequence of Eq.~\eqref{eq:gammadef} applied to product states of definite mode occupation. Pulling back through the network isometry $\widetilde V$, an $n$-photon input state $|\psi\rangle\in\Sym^n(\mathfrak h)$ is mapped to $\widetilde V^{\otimes n}|\psi\rangle\in\Sym^n(\mathfrak h_{\rm out})$, so
\begin{align}
  \langle\psi|N_S^{(n)}|\phi\rangle
  &=\gamma_S\,\langle\psi|(\widetilde V^\dagger)^{\otimes n} D_S^{\otimes n}\widetilde V^{\otimes n}|\phi\rangle \nonumber\\
  &=\gamma_S\,\langle\psi|\big(\widetilde V^\dagger D_S \widetilde V\big)^{\otimes n}|\phi\rangle ,
\end{align}
and $\widetilde V^\dagger D_S\widetilde V=\widetilde V^\dagger \widetilde V-\sum_{d\in S}\eta_d v_d^{\phantom\dagger}v_d^\dagger=A_S$ by isometry. Positivity $A_S\succeq 0$ follows from $\sum_{d\in S}\eta_d v_dv_d^\dagger\preceq\sum_d v_dv_d^\dagger\preceq\id$, Eq.~\eqref{eq:contraction}, and $A_S\preceq\id$ is manifest.
\end{proof}

Lemma~\ref{lem:gamma} is the structural heart of the paper: it says that every silence operator, restricted to $n$ photons, is a \emph{symmetric tensor power of an explicit $M\times M$ matrix}. All spectral questions about click-pattern operators thereby become questions about tensor powers of small matrices. The model also covers multimode detectors: if detector $d$ collects several output modes $\mu\in\mathcal M_d$ (temporal, spectral, or spatial), possibly with mode-dependent efficiencies $\eta_{d\mu}$, its silence operator is still a product of number-diagonal attenuations, and every statement below holds verbatim upon replacing $\eta_d v_d^{\phantom\dagger}v_d^\dagger$ by $\sum_{\mu\in\mathcal M_d}\eta_{d\mu}v_{d\mu}^{\phantom\dagger}v_{d\mu}^\dagger$; this covers, in particular, the spatial-mode-dependent mismatch models of Refs.~\cite{Sajeed2015,Zhang2017}.

\subsection{The cross-click operator}

Two families of silence operators matter. For each basis $b\in\{1,\dots,m\}$ define
\begin{align}
  Q_b &:= N_{\bar{\mathbb B}_b} = \gamma_b\,\Gamma(A_b), &
  A_b &= \id-\!\!\sum_{d\notin\mathbb{B}_b}\!\eta_d v_d^{\phantom\dagger}v_d^\dagger,
  \label{eq:Qb}\\
  F_\emptyset &:= N_{\{1,\dots,D\}} = \gamma_\emptyset\,\Gamma(A_\emptyset), &
  A_\emptyset &= \id-\sum_{d}\eta_d v_d^{\phantom\dagger}v_d^\dagger,
  \label{eq:Fempty}
\end{align}
with $\bar{\mathbb B}_b$ the complement of $\mathbb B_b$, $\gamma_b:=\gamma_{\bar{\mathbb B}_b}$ and $\gamma_\emptyset:=\prod_d(1-\varepsilon_d)$. Thus $Q_b$ is the probability operator of ``all clicks (if any) are confined to basis $b$'', and $F_\emptyset$ that of ``no clicks at all''. It is convenient to introduce the positive \emph{coverage operators}
\begin{equation}\label{eq:Rb}
  R_b:=\sum_{d\in\mathbb{B}_b}\eta_d\, v_d^{\phantom\dagger}v_d^\dagger\;\succeq\;0,
\end{equation}
in terms of which
\begin{equation}\label{eq:AbRb}
  \id - A_b=\sum_{b'\neq b}R_{b'},\qquad
  \id - A_\emptyset=\sum_{b'}R_{b'},\qquad
  A_b = A_\emptyset + R_b .
\end{equation}

A \emph{cross click} is the event that detectors in at least two different bases fire. Its POVM element is obtained by subtracting from the identity all (mutually exclusive) click patterns confined to a single basis. Since the events ``clicks confined to $b$'' for different $b$ overlap exactly in the no-click pattern, inclusion--exclusion gives
\begin{equation}\label{eq:FC}
  \boxed{\;F_C \;=\; \id-\sum_{b=1}^m Q_b + (m-1)\,F_\emptyset\;}
\end{equation}
For the six-state analyzer ($m=3$) this reproduces the definition used in Refs.~\cite{Zhang2017,YeThesis},
$F_C=\id-F_H-F_V-F_D-F_A-F_R-F_L-F_{HV}-F_{DA}-F_{RL}-F_\emptyset$, because $Q_{\mathbb Z}=F_\emptyset+F_H+F_V+F_{HV}$ etc.

We denote the complementary \emph{no-cross-click} operator by
\begin{equation}\label{eq:NC}
  N_C:=\id-F_C=\sum_b Q_b-(m-1)F_\emptyset ,
\end{equation}
which, being the probability operator of an event, satisfies $0\preceq N_C\preceq \id$, and whose $n$-photon block is, by Lemma~\ref{lem:gamma},
\begin{equation}\label{eq:NCn}
  N_C^{(n)}=\Big[\sum_b \gamma_b\,A_b^{\otimes n}-(m-1)\,\gamma_\emptyset\,A_\emptyset^{\otimes n}\Big]\Big|_{\Sym^n(\mathfrak h)} .
\end{equation}
The central quantity is
\begin{equation}\label{eq:fn}
  f^{(n)} := \eigmin\big(F_C^{(n)}\big)=1-\big\lVert N_C^{(n)}\big\rVert ,
\end{equation}
where $\lVert\cdot\rVert$ is the operator norm; the second equality uses $N_C^{(n)}\succeq0$.

\section{Main results}\label{sec:main}

Throughout this section, ``analyzer'' means any device of the class defined in Sec.~\ref{sec:setup}: arbitrary $M$, $D$, $m\ge2$, network (lossy or not), efficiencies $\eta_d\in(0,1]$, and dark-count probabilities $\varepsilon_d\in[0,1)$. We will repeatedly use the elementary fact that positive semidefinite ordering is preserved by tensor powers:

\begin{lemma}\label{lem:tensorder}
If $0\preceq X\preceq Y$ on $\mathfrak h$, then $0\preceq X^{\otimes n}\preceq Y^{\otimes n}$ on $\mathfrak h^{\otimes n}$ for every $n\ge1$.
\end{lemma}
\begin{proof}
Telescoping, $Y^{\otimes n}-X^{\otimes n}=\sum_{k=1}^{n} Y^{\otimes(k-1)}\otimes(Y-X)\otimes X^{\otimes(n-k)}$, and each summand is a tensor product of positive semidefinite factors, hence positive semidefinite.
\end{proof}

Note that Lemma~\ref{lem:tensorder} is special to the \emph{positive semidefinite} setting; for general Hermitian operators, $X\preceq Y$ does not imply $X^{\otimes n}\preceq Y^{\otimes n}$. Since $A_\emptyset\succeq0$ and $A_b=A_\emptyset+R_b\succeq A_\emptyset$ by Eq.~\eqref{eq:AbRb}, Lemma~\ref{lem:tensorder} yields
\begin{equation}\label{eq:powerorder}
  0\;\preceq\;A_\emptyset^{\otimes n}\;\preceq\;A_b^{\otimes n}\qquad\text{for every } b,\ n .
\end{equation}
Together with $\gamma_b\ge\gamma_\emptyset$ this already shows $Q_b\succeq F_\emptyset$ sector by sector, so each term $Q_b-F_\emptyset$ in $N_C=F_\emptyset+\sum_b(Q_b-F_\emptyset)$ is positive---the operator-level counterpart of the fact that the underlying events are nested.

\subsection{Monotonicity}

\begin{theorem}[Monotonicity]\label{thm:mono}
For every analyzer and every $n\ge0$,
\begin{equation}
  f^{(n+1)}\;\ge\;f^{(n)} ,
\end{equation}
with the convention $A^{\otimes0}:=1$ on the one-dimensional vacuum sector, so that $N_C^{(0)}=\sum_b\gamma_b-(m-1)\gamma_\emptyset$ is a scalar.
\end{theorem}

\begin{proof}
Define on the \emph{full} tensor space $\mathfrak h^{\otimes n}$ the operator
\begin{equation}
  \widetilde N^{(n)}:=\sum_b\gamma_b A_b^{\otimes n}-(m-1)\gamma_\emptyset A_\emptyset^{\otimes n},
\end{equation}
whose restriction to $\Sym^n(\mathfrak h)$ is $N_C^{(n)}$ by Eq.~\eqref{eq:NCn}. The key step is the operator inequality
\begin{equation}\label{eq:keyineq}
  \widetilde N^{(n+1)}\;\preceq\;\widetilde N^{(n)}\otimes\id_\mathfrak{h}.
\end{equation}
To prove it, compute the difference and use Eq.~\eqref{eq:AbRb}:
\begin{align}
  \widetilde N^{(n)}\otimes\id-\widetilde N^{(n+1)}
  &=\sum_b\gamma_b\,A_b^{\otimes n}\otimes(\id-A_b)\nonumber\\
  &\quad-(m-1)\gamma_\emptyset\,A_\emptyset^{\otimes n}\otimes(\id-A_\emptyset)\nonumber\\
  &=\sum_b\gamma_b\,A_b^{\otimes n}\otimes\!\sum_{b'\neq b}\!R_{b'}\nonumber\\
  &\quad-(m-1)\gamma_\emptyset\,A_\emptyset^{\otimes n}\otimes\!\sum_{b'}\!R_{b'}\nonumber\\
  &=\sum_{b'}\Big[\underbrace{\sum_{b\neq b'}\!\big(\gamma_b A_b^{\otimes n}-\gamma_\emptyset A_\emptyset^{\otimes n}\big)}_{\succeq\,0 \text{ by Eq.~\eqref{eq:powerorder}, }\gamma_b\ge\gamma_\emptyset}\Big]\otimes \underbrace{R_{b'}}_{\succeq 0}\nonumber\\
  &\succeq 0,
\end{align}
where in the last rewriting we used that the bracket contains exactly $m-1$ terms. This proves Eq.~\eqref{eq:keyineq}.

Now let $|\psi\rangle\in\Sym^{n+1}(\mathfrak h)$ be a unit vector achieving $\lVert N_C^{(n+1)}\rVert$. Since $\Sym^{n+1}(\mathfrak h)\subset\Sym^{n}(\mathfrak h)\otimes\mathfrak h$, the reduced state $\rho:=\Tr_{n+1}|\psi\rangle\langle\psi|$ (partial trace over the last tensor factor) is a density operator supported on $\Sym^n(\mathfrak h)$ [writing $|\psi\rangle=\sum_i|\varphi_i\rangle\otimes|e_i\rangle$ with $|\varphi_i\rangle\in\Sym^n(\mathfrak h)$ and $\{|e_i\rangle\}$ an orthonormal basis of $\mathfrak h$ gives $\rho=\sum_i|\varphi_i\rangle\langle\varphi_i|$]. Then
\begin{align}
  1-f^{(n+1)}&=\langle\psi|\widetilde N^{(n+1)}|\psi\rangle
  \;\le\;\langle\psi|\,\widetilde N^{(n)}\otimes\id\,|\psi\rangle\nonumber\\
  &=\Tr\big[\widetilde N^{(n)}\rho\big]
  =\Tr\big[N_C^{(n)}\rho\big]\nonumber\\
  &\le\;\big\lVert N_C^{(n)}\big\rVert=1-f^{(n)},
\end{align}
using Eq.~\eqref{eq:keyineq} in the first inequality, the support of $\rho$ in the middle equality, and $N_C^{(n)}\succeq0$ in the last inequality.
\end{proof}

\begin{remark}
The proof formalizes the physical intuition that ``removing one photon cannot create a cross click'': conditioning the $(n{+}1)$-photon no-cross-click probability on the fate of one photon can only be dominated by the best $n$-photon value. Note that no assumption on the polarization state, on efficiency symmetry, or on the splitting ratios entered anywhere; monotonicity is a structural property of passive analyzers with threshold detectors, including dark counts.
\end{remark}

\subsection{Two-sided bounds and the exact convergence rate}

\begin{theorem}[Exponential bounds]\label{thm:rate}
For every analyzer and every $n\ge1$,
\begin{equation}\label{eq:twosided}
  \max_b\;\gamma_b\,\lVert A_b\rVert^{\,n}
  \;\le\;
  1-f^{(n)}
  \;\le\;
  \sum_b\;\gamma_b\,\lVert A_b\rVert^{\,n},
\end{equation}
with
\begin{equation}\label{eq:normAb}
  \lVert A_b\rVert
  \;=\;1-\eigmin\Big(\sum_{b'\neq b}R_{b'}\Big)
  \;=\;1-\eigmin\Big(\sum_{d\notin\mathbb B_b}\eta_d\,v_d^{\phantom\dagger}v_d^\dagger\Big).
\end{equation}
Consequently
\begin{equation}\label{eq:rate}
  \lim_{n\to\infty}\big(1-f^{(n)}\big)^{1/n}\;=\;\max_b\,\lVert A_b\rVert,
\end{equation}
and $f^{(n)}\to1$ (exponentially) if and only if the \emph{complementary coverage condition}
\begin{equation}\label{eq:coverage}
  \sum_{b'\neq b} R_{b'} \;\succ\;0\qquad\text{for every }b
\end{equation}
holds, i.e., the detectors \emph{outside} any single basis jointly monitor every input mode with nonzero efficiency.
\end{theorem}

\begin{proof}
\emph{Lower bound.} By Eq.~\eqref{eq:powerorder} and $\gamma_{b'}\ge\gamma_\emptyset$ we have $N_C=Q_b+\sum_{b'\neq b}(Q_{b'}-F_\emptyset)\succeq Q_b$ sector by sector, for each $b$. Let $u_b$ be a normalized top eigenvector of $A_b$; then $u_b^{\otimes n}\in\Sym^n(\mathfrak h)$ and
\begin{equation}
  1-f^{(n)}\ge\langle u_b^{\otimes n}|Q_b^{(n)}|u_b^{\otimes n}\rangle=\gamma_b\lVert A_b\rVert^n .
\end{equation}

\emph{Upper bound.} $N_C\preceq\sum_bQ_b$ since $F_\emptyset\succeq0$, so by the triangle inequality and $\lVert A_b^{\otimes n}|_{\Sym^n}\rVert=\lVert A_b\rVert^n$ (the top eigenvector $u_b^{\otimes n}$ is symmetric),
\begin{equation}
  1-f^{(n)}\le\sum_b\gamma_b\lVert A_b\rVert^n .
\end{equation}
Equation~\eqref{eq:normAb} follows from Eq.~\eqref{eq:AbRb}, Eq.~\eqref{eq:rate} from the squeeze between the two bounds (the upper bound is at most $m\max_b\gamma_b\lVert A_b\rVert^n$), and the coverage criterion from $\lVert A_b\rVert<1\Leftrightarrow\sum_{b'\neq b}R_{b'}\succ0$.
\end{proof}

Theorem~\ref{thm:rate} shows that $1-f^{(n)}$ is pinned between two explicit exponentials that differ only by the constant factor $m$; the \emph{rate} of the exponential convergence---the operationally decisive quantity in the weight bounds of Sec.~\ref{sec:application}---is identified \emph{exactly}, through the explicit single-photon quantity \eqref{eq:normAb} [in closed form for the six-state analyzer, Eq.~\eqref{eq:sixrate}]. Physically, $\langle u|A_b|u\rangle$ is the probability that a single photon of polarization $u$ evades every detector outside basis $b$; a cross click can fail to occur only if each photon evades all bases but one, and $\lVert A_b\rVert^{n}$ is the corresponding best-case no-click weight.

Two closed-form specializations are worth recording.

\begin{corollary}[Ideal detectors]\label{cor:ideal}
Suppose the network is lossless, $\varepsilon_d=0$, $\eta_d=1$ for all $d$, and each basis resolves a complete measurement, $\sum_{d\in\mathbb B_b}v_dv_d^\dagger=p_b\id$ with splitting probabilities $p_b>0$, $\sum_bp_b=1$. Then $A_b=p_b\id$, $A_\emptyset=0$, and for every $n\ge1$ the $n$-photon block is exactly a multiple of the identity,
\begin{equation}\label{eq:ideal}
  N_C^{(n)}=\Big(\sum_b p_b^{\,n}\Big)\id,
  \qquad
  f^{(n)}=1-\sum_{b=1}^m p_b^{\,n} .
\end{equation}
For the balanced passive six-state analyzer, $f^{(n)}=1-3^{\,1-n}$.
\end{corollary}

\begin{proof}
Immediate from Eqs.~\eqref{eq:AbRb} and \eqref{eq:NCn}; $\sum_bp_b=1$ gives $A_b=\id-\sum_{b'\ne b}p_{b'}\id=p_b\id$, and $A_\emptyset^{\otimes n}=0$ for $n\ge1$.
\end{proof}

Note $f^{(1)}=0$, as it must be: a single photon can trigger at most one detector. Corollary~\ref{cor:ideal} also exposes the \emph{mechanism} behind Eq.~\eqref{eq:ideal}: with ideal detectors the no-cross-click probability of any $n$-photon state is simply the (state-independent) probability $\sum_b p_b^n$ that the polarization-independent splitter routes all $n$ photons into the same arm.

\begin{corollary}[Six-state analyzer with mismatched efficiencies]\label{cor:sixstate}
For the balanced ($p_b=\tfrac13$) six-state polarization analyzer with efficiencies $\eta_{b1},\eta_{b2}$ in basis $b$ and no internal loss,
\begin{equation}\label{eq:sixrate}
  \lVert A_b\rVert=1-\frac{1}{3}\!\!\sum_{b'\neq b}\!\bar\eta_{b'}
  +\frac{1}{6}\sqrt{\sum_{b'\neq b}\big(\eta_{b'1}-\eta_{b'2}\big)^{2}},
\end{equation}
where $\bar\eta_{b'}=(\eta_{b'1}+\eta_{b'2})/2$.
\end{corollary}

\begin{proof}
See Appendix~\ref{app:sixstate}: the three basis projectors define mutually orthogonal Pauli axes on the Bloch sphere, so $\sum_{b'\neq b}R_{b'}=c\,\id+\vec r\cdot\vec\sigma$ with $c=\frac13\sum_{b'\neq b}\bar\eta_{b'}$ and $|\vec r|=\frac16[\sum_{b'\neq b}(\Delta\eta_{b'})^2]^{1/2}$.
\end{proof}

Equation~\eqref{eq:sixrate} makes the effect of efficiency mismatch on the certification power of cross clicks completely explicit: mean efficiency loss enters through $\bar\eta$, mismatch through the Euclidean norm of the per-basis imbalances, and the two contributions are additive at the level of the rate.

\subsection{The joint two-party operator}

In entanglement verification and entanglement-based QKD \cite{Ekert1991,BBM92}, both Alice and Bob operate analyzers of the above type, and the relevant object---used in the one-sided form $\id_A\otimes F_C^B$ in Ref.~\cite{Zhang2017} and in the two-sided form below in Ref.~\cite{YeThesis}---is the \emph{joint cross-click operator}
\begin{equation}\label{eq:CC}
  C_C:=\id_{AB}-\big(\id_A-F_C^A\big)\otimes\big(\id_B-F_C^B\big)
  =\id_{AB}-N_C^A\otimes N_C^B,
\end{equation}
the probability operator of ``at least one party registers a cross click''. Its restriction to the $(n_A,n_B)$-photon sector defines
\begin{equation}
  f^{(n_A,n_B)}:=\eigmin\big(C_C^{(n_A,n_B)}\big).
\end{equation}

\begin{theorem}[Exact factorization]\label{thm:joint}
For all $n_A,n_B\ge0$,
\begin{equation}\label{eq:factorization}
  1-f^{(n_A,n_B)}\;=\;\big(1-f_A^{(n_A)}\big)\big(1-f_B^{(n_B)}\big).
\end{equation}
Consequently, $f^{(n_A,n_B)}$ is (i) nondecreasing in $n_A$ and in $n_B$ separately, (ii) jointly convergent to $1$ whenever either party's analyzer satisfies the coverage condition \eqref{eq:coverage}, and (iii) computable from the two single-party curves. The same holds, mutatis mutandis, for $K$ parties: $1-f^{(n_1,\dots,n_K)}=\prod_k(1-f_k^{(n_k)})$.
\end{theorem}

\begin{proof}
$N_C^A\otimes N_C^B$ is a tensor product of positive semidefinite operators, and the operator norm is multiplicative on tensor products of positive operators: $\lVert X\otimes Y\rVert=\lVert X\rVert\,\lVert Y\rVert$. Hence
$1-f^{(n_A,n_B)}=\lVert (N_C^A)^{(n_A)}\otimes (N_C^B)^{(n_B)}\rVert
=\lVert (N_C^A)^{(n_A)}\rVert\,\lVert (N_C^B)^{(n_B)}\rVert$.
Claims (i)--(iii) follow from Theorems~\ref{thm:mono} and \ref{thm:rate}.
\end{proof}

Theorem~\ref{thm:joint} replaces the two-dimensional numerical scans of Ref.~\cite{YeThesis} by a product of two one-dimensional curves, each of which is itself pinned down by Theorem~\ref{thm:rate}. It also explains, as an exact statement rather than a numerical observation, the degeneracies visible in those scans: since a single photon cannot produce a cross click, $f_B^{(0)}=f_B^{(1)}=0$ (absent dark counts), hence
\begin{equation}
  f^{(n_A,0)}=f^{(n_A,1)}=f_A^{(n_A)},
\end{equation}
i.e., the $n_B=0$ and $n_B=1$ curves coincide identically.

\begin{remark}[Relation to PPT-constrained quantities]
Ref.~\cite{Zhang2017} works with the constrained minimum $c^{(n)}_{\min}=\min\Tr[\rho_{AB}\,\id_A\otimes G_C]$ over bipartite states with positive partial transpose (PPT). For this one-sided operator, however, the unconstrained minimum $f^{(n)}$ is attained on product states, which are PPT; hence $c^{(n)}_{\min}=f^{(n)}$ exactly, and Theorems~\ref{thm:mono} and \ref{thm:rate} settle the monotonicity and convergence of the cross-click quantity of Ref.~\cite{Zhang2017}---posed there explicitly as an open problem---rather than merely bounding it. The double-click and effective-error quantities of the active-detection scheme in that work have a different structure and are not covered by the present analysis.
\end{remark}

\section{Application: closed-form photon-number weight bounds}\label{sec:application}

We now spell out how the theorems upgrade the security-proof step in which cross clicks are used, following the protocol structure of Refs.~\cite{Zhang2021,YeThesis} (entanglement-based four--six-state QKD with passive detection, as demonstrated experimentally in Ref.~\cite{Tannous2019}): an untrusted source distributes bipartite states; Alice and Bob measure with characterized passive analyzers; announced outcomes include the cross-click flag. Because all POVM elements are photon-number block diagonal, dephasing $\rho_{AB}$ in local photon number changes no observed statistic; for the weight-bound step considered here one may therefore take the shared state block diagonal \cite{Zhang2017,YeThesis}
\begin{equation}
  \rho_{AB}=\bigoplus_{n_A,n_B\ge0}P_{n_A,n_B}\,\rho^{(n_A,n_B)},
\end{equation}
with $\{P_{n_A,n_B}\}$ an unknown probability distribution over photon-number sectors. The observed joint cross-click probability is
\begin{equation}
  q_C=\Tr[\rho_{AB}C_C]=\sum_{n_A,n_B}P_{n_A,n_B}\Tr\big[\rho^{(n_A,n_B)}C_C^{(n_A,n_B)}\big] .
\end{equation}

\begin{proposition}[Low-photon-number block weight]\label{prop:weight}
Let $P_{\le1,\le1}:=\sum_{n_A\le1,n_B\le1}P_{n_A,n_B}$ be the weight of the block with at most one photon on each side, and suppose $\min\{f_A^{(2)},f_B^{(2)}\}>0$ (which holds, e.g., for splitter-type analyzers with all $\eta_d>0$; see Appendix~\ref{app:smalln}). Then
\begin{equation}\label{eq:qubitweight}
  P_{\le1,\le1}\;\ge\;1-\frac{q_C}{\min\{f_A^{(2)},\,f_B^{(2)}\}},
\end{equation}
and each $f^{(2)}$ is bounded explicitly by Theorem~\ref{thm:rate},
\begin{equation}\label{eq:f2bound}
  f^{(2)}\;\ge\;1-\sum_b\gamma_b\lVert A_b\rVert^{2},
\end{equation}
with $\lVert A_b\rVert$ from Eq.~\eqref{eq:normAb} [Eq.~\eqref{eq:sixrate} for the six-state analyzer]; alternatively, $f^{(2)}=1-\eigmax\big(N_C^{(2)}\big)$ is obtained by diagonalizing the explicit matrix \eqref{eq:NCn} of dimension $\dim\Sym^2(\mathbb{C}^M)=M(M{+}1)/2$ [$3\times3$ for polarization qubits]. The surrogate \eqref{eq:f2bound} is informative only when its right-hand side is positive; otherwise one uses the exact $f^{(2)}$, whose strict positivity for splitter-type analyzers is established in Appendix~\ref{app:smalln}.
\end{proposition}

\begin{proof}
Outside this block, $n_A\ge2$ or $n_B\ge2$. By Theorems~\ref{thm:mono} and \ref{thm:joint}, $f^{(n_A,n_B)}\ge\min\{f_A^{(2)},f_B^{(2)}\}$ on all such sectors (monotonicity in each argument, and $1-f^{(n_A,n_B)}\le 1-f_A^{(n_A)}\le 1-f_A^{(2)}$ when $n_A\ge2$, similarly for $B$). Dropping the nonnegative contribution of the retained block,
\begin{align}
  q_C&\ge\sum_{(n_A,n_B)\notin\{0,1\}^2}\!\!P_{n_A,n_B}f^{(n_A,n_B)}\nonumber\\
  &\ge(1-P_{\le1,\le1})\min\{f_A^{(2)},f_B^{(2)}\}. \qedhere
\end{align}
\end{proof}

More generally, for any truncation thresholds $(N_A,N_B)$,
\begin{equation}\label{eq:generalweight}
  \sum_{n_A\le N_A,\,n_B\le N_B}\!\!\!\!P_{n_A,n_B}
  \;\ge\;1-\frac{q_C}{\min\{f_A^{(N_A+1)},\,f_B^{(N_B+1)}\}},
\end{equation}
valid whenever the denominator is positive, with $f^{(N+1)}\ge1-\sum_b\gamma_b\lVert A_b\rVert^{N+1}$ in explicit form; by Theorem~\ref{thm:rate} the closed-form surrogate is positive for all sufficiently large $N_A,N_B$ under the coverage condition \eqref{eq:coverage}. Such weight bounds are precisely the input required by dimension-reduction methods for QKD \cite{Upadhyaya2021} and by the numerical key-rate framework of Refs.~\cite{Coles2016,Winick2018}, where the infinite-dimensional optimization is replaced by an optimization over the retained block plus a penalty controlled by the discarded weight.

We stress the qualitative upgrade relative to the cross-click analyses of Refs.~\cite{Zhang2017,Zhang2021,YeThesis}. First, \emph{rigor}: Eqs.~\eqref{eq:qubitweight}--\eqref{eq:generalweight} rest on Theorems~\ref{thm:mono}--\ref{thm:joint}, which hold for all $n$; in those works, monotonicity beyond the numerically checked window ($n\le20$ in Ref.~\cite{Zhang2017}, $n\le60$ in Ref.~\cite{YeThesis}) was an extrapolation, so the resulting statements were, strictly speaking, conditional. Second, \emph{tightness control}: the two-sided bound \eqref{eq:twosided} shows that the closed-form estimate \eqref{eq:f2bound} is at worst a factor $m$ off in $1-f$, and the exact $f^{(2)}$ is available from a $3\times3$ diagonalization. Third, \emph{cost}: no Fock-space construction is needed at all; where sector-wise exact values are nevertheless desired (e.g., to tighten $f^{(N+1)}$ for large $N$), Sec.~\ref{sec:numerics} provides a polynomial-time per-sector algorithm, replacing the exponentially costly naive constructions.

The relation to the flag-state squashing bounds of Ref.~\cite{Kamin2024} is complementary. That work derives analytic \emph{one-sided} lower bounds on the probability of detectable events for arbitrary passive setups, which is what prepare-and-measure subspace estimation requires. The present results deliver, for the specific but widely used cross-click event, the exact minimum eigenvalue structure: monotonicity of $f^{(n)}$ itself (the quantity constrained in the entanglement-based analyses \cite{Zhang2017,YeThesis}), matching upper bounds that certify how much room for improvement remains, and the exact bipartite factorization \eqref{eq:factorization}, which has no analog in the single-receiver setting.

\section{Numerical illustration and cross-checks}\label{sec:numerics}

\begin{figure*}[tb]
\includegraphics[width=0.98\textwidth]{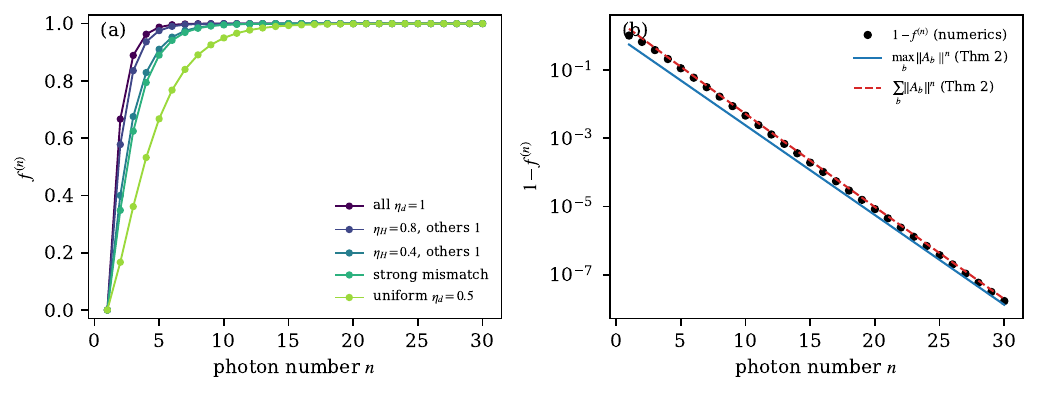}
\caption{(a) Minimum eigenvalue $f^{(n)}$ of the cross-click operator of the balanced passive six-state analyzer on the $n$-photon sector, for several detector-efficiency configurations, computed exactly by the symmetric-power algorithm of Sec.~\ref{sec:numerics}. All curves are monotone, as guaranteed by Theorem~\ref{thm:mono}; for ideal detectors the curve is exactly $f^{(n)}=1-3^{1-n}$ (Corollary~\ref{cor:ideal}). (b) Convergence $1-f^{(n)}$ for the strongly mismatched configuration $\vec\eta=(0.9,0.75,1.0,0.6,0.85,0.7)$ on a logarithmic scale, together with the two-sided bounds of Theorem~\ref{thm:rate}. The exact values are squeezed between $\max_b\lVert A_b\rVert^n$ and $\sum_b\lVert A_b\rVert^n$ and follow the closed-form rate $\max_b\lVert A_b\rVert=0.5462$ [Eq.~\eqref{eq:sixrate}].}
\label{fig:fn}
\end{figure*}

\begin{figure}[tb]
\includegraphics[width=0.95\columnwidth]{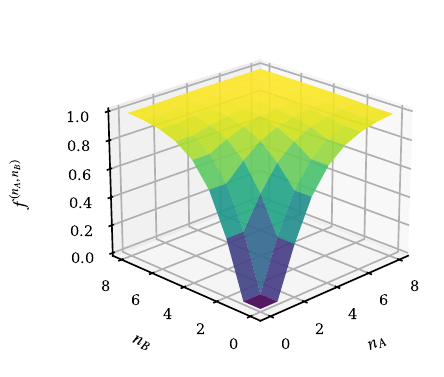}
\caption{Joint minimum eigenvalue $f^{(n_A,n_B)}$ for two different mismatched six-state analyzers, evaluated from the exact factorization of Theorem~\ref{thm:joint}, $1-f^{(n_A,n_B)}=(1-f_A^{(n_A)})(1-f_B^{(n_B)})$. The surface reproduces the structure obtained by direct numerical diagonalization in Ref.~\cite{YeThesis}, including the exact coincidence of the $n_B=0$ and $n_B=1$ slices.}
\label{fig:joint}
\end{figure}

Although the theorems remove the need for numerics in applications, numerical evaluation remains useful for illustration and for validating the formalism. Lemma~\ref{lem:gamma} itself yields an efficient algorithm: on $\Sym^n(\mathbb{C}^M)$, the block $A^{\otimes n}|_{\Sym^n}$ is the $n$-th symmetric power of an $M\times M$ matrix, computable directly in the occupation-number basis---for $M=2$, an $(n{+}1)\times(n{+}1)$ matrix obtained from the action of $A$ on homogeneous polynomials of degree $n$---at cost polynomial in $n$. This replaces the direct Fock-space construction of the POVM elements from products of creation operators, whose cost grows exponentially with $n$ in the naive labeled-mode expansion (the sector dimension itself grows only polynomially at fixed mode number) and which required dedicated caching strategies to reach $n=60$ in Ref.~\cite{YeThesis}; the symmetric-power route evaluates thousands of sectors in seconds.

We validated the formalism in two independent ways. First, for $n\le4$ and strongly mismatched efficiencies we constructed the silence operators $Q_b^{(n)}$ by brute force in the six-mode output Fock space---expanding $n$-photon input monomials through the network isometry and applying the diagonal no-click weights $\prod_d(1-\eta_d)^{n_d}$---and verified that their spectra coincide (to machine precision) with those of $\gamma_bA_b^{\otimes n}|_{\Sym^n}$ from Lemma~\ref{lem:gamma}. Second, we verified numerically, for the same configurations, the key operator inequality \eqref{eq:keyineq}, the two-sided bounds \eqref{eq:twosided}, the ideal-detector formula \eqref{eq:ideal}, the closed-form rate \eqref{eq:sixrate}, the factorization \eqref{eq:factorization}, and monotonicity in the presence of dark counts.

Figure~\ref{fig:fn}(a) shows $f^{(n)}$ for the six-state analyzer under several efficiency configurations, including the strongly mismatched set $\vec\eta=(\eta_H,\eta_V,\eta_D,\eta_A,\eta_R,\eta_L)=(0.9,0.75,1.0,0.6,0.85,0.7)$. Figure~\ref{fig:fn}(b) displays $1-f^{(n)}$ on a logarithmic scale together with the bounds of Theorem~\ref{thm:rate}; the exact values track the closed-form rate $\max_b\lVert A_b\rVert=0.5462$ [Eq.~\eqref{eq:sixrate} gives $\lVert A_{\mathbb Z}\rVert=0.5462$, $\lVert A_{\mathbb X}\rVert=0.5020$, $\lVert A_{\mathbb Y}\rVert=0.5295$ for this configuration]. Figure~\ref{fig:joint} shows the joint surface $f^{(n_A,n_B)}$ for two different analyzers, evaluated from the factorization \eqref{eq:factorization}; it reproduces the structure previously obtained by two-dimensional numerical scans, now as an exact formula.

\section{Discussion}\label{sec:discussion}

We have given a comprehensive analytic account of the spectral growth of cross-click operators for passive multi-basis photodetection with threshold detectors: monotonicity for all photon numbers, two-sided exponential bounds with a closed-form rate, exact values for ideal detectors, and an exact factorization for multipartite joint operators. The results hold for arbitrary passive analyzers---any number of input modes, bases, and detectors, arbitrary efficiency mismatch, internal loss, and dark counts---and thereby cover, beyond the polarization qubit receivers that motivated the question \cite{Zhang2017,Zhang2021,YeThesis}, time-bin and spatial-mode analyzers of arbitrary dimension.

Conceptually, the analysis rests on a single structural fact (Lemma~\ref{lem:gamma}): silence operators of threshold detectors behind passive linear optics are second quantizations of explicit single-photon contractions. This ``one contraction per silence event'' viewpoint may prove useful beyond the present problem, e.g., in the analysis of double-click operators for squashing models \cite{Beaudry2008,Gittsovich2014} and of the detectable-event bounds underlying flag-state squashing \cite{LiLutkenhaus2020,Kamin2024}, in detector self-characterization, and in photon-number-resolved state reconstruction, where the same tensor-power structure controls the sector-wise spectra.

For security proofs, the practical consequences are threefold: the photon-number weight bounds that feed dimension-reduction and numerical key-rate methods \cite{Coles2016,Winick2018,Upadhyaya2021} become explicit and rigorous for all $n$; the exact convergence rate quantifies how efficiency mismatch degrades the certification power of cross clicks [Eq.~\eqref{eq:sixrate}]; and the joint two-party analysis collapses to single-party curves. In finite-key analyses \cite{Tomamichel2012}, where $q_C$ is estimated statistically, the same monotone bounds apply verbatim to the estimated value with its confidence interval.

Several extensions merit further work. First, our detector model assumes independent dark counts and uncorrelated efficiencies; afterpulsing and other temporally correlated noise break the product structure of Lemma~\ref{lem:gamma}, and it would be interesting to identify the weakest noise assumptions under which monotonicity survives. Second, the lower bound in Theorem~\ref{thm:rate} is witnessed by the symmetric product states $u_b^{\otimes n}$; characterizing the exact minimizers of $F_C^{(n)}$ at finite $n$, and the finite-$n$ tightness of the explicit bounds, would sharpen the weight bounds further. Third, the same technique should yield analytic monotonicity statements for the double-click and effective-error quantities of Ref.~\cite{Zhang2017} (for which the PPT constraint is active), for basis-dependent (active-choice) receivers, and for cross-click-type operators in measurement-device-independent settings \cite{Lo2012,Braunstein2012}. Finally, it is natural to ask for the optimal certified truncation $(N_A,N_B)$ as a function of observed data and target key rate, now that all ingredients of Eq.~\eqref{eq:generalweight} are analytic.

\begin{acknowledgments}
The problem treated here was first encountered in the four--six-state experiment of Ref.~\cite{Tannous2019} and formulated in the author's doctoral work, completed in 2020 and reported in Ref.~\cite{YeThesis}, where the properties proved here were verified numerically. The author thanks Norbert L\"utkenhaus for earlier discussions at the Institute for Quantum Computing (IQC) and for supervision of the extension work of Ref.~\cite{Tannous2019}, and colleagues at IQC, in particular Ashutosh Marwah, Jie Lin, and Nicky Kai Hong Li, and at the Department of Physics, Tsinghua University, for discussions and help.

\end{acknowledgments}

\section*{Statement on the use of AI tools}
Large language models were used substantially in the preparation of this work, including the assistance with the analytic proofs, the verification code, the figures, and the manuscript text; a second model was used for additional checking of the mathematics and of the literature. The scripts that reproduce all numerical checks and figures are provided as Supplemental Material.

\appendix

\section{Details on Lemma~\ref{lem:gamma}}\label{app:lemma}

We record the two standard identities used in the proof of Lemma~\ref{lem:gamma}, for completeness.

(i) \emph{No-click operator of a lossy threshold detector.} A threshold detector of unit efficiency monitoring a mode with number operator $\hat n$ reports ``no click'' on the vacuum only; preceded by a beam splitter of transmittance $\eta$ (the standard efficiency model), the no-click probability on a Fock state $|k\rangle$ is $(1-\eta)^k$, since each photon is independently reflected into the loss port with probability $1-\eta$ and detected otherwise. Because the operator is diagonal in the Fock basis with these eigenvalues, it equals $(1-\eta)^{\hat n}$. An independent dark count with probability $\varepsilon$ multiplies the no-click operator by the scalar $(1-\varepsilon)$.

(ii) \emph{Products of number-diagonal attenuations are second quantizations.} For $x_d\in[0,1]$,
\begin{equation}
  \prod_d x_d^{\hat n_d}=\Gamma\Big(\bigoplus_d x_d\Big),
\end{equation}
because both sides are diagonal in the multimode Fock basis with eigenvalue $\prod_d x_d^{k_d}$ on $|k_1,k_2,\dots\rangle$, which is precisely the action of $(\bigoplus_dx_d)^{\otimes n}$ on the corresponding symmetrized $n$-photon vector, $n=\sum_dk_d$.

Finally, the pull-back step uses the identity
\begin{equation}\label{eq:pullback}
  \Gamma(\widetilde V)^\dagger\,\Gamma_{\rm out}(D)\,\Gamma(\widetilde V)
  \;=\;\Gamma\big(\widetilde V^\dagger D\widetilde V\big),
\end{equation}
which on each $n$-photon sector reads $(\widetilde V^\dagger)^{\otimes n}D^{\otimes n}\widetilde V^{\otimes n}=(\widetilde V^\dagger D\widetilde V)^{\otimes n}$ and holds because tensor products factor; this is all that enters POVM elements on the input space. [We caution that the stronger intertwining $\Gamma_{\rm out}(D)\Gamma(\widetilde V)=\Gamma(\widetilde V)\Gamma(\widetilde V^\dagger D\widetilde V)$ is \emph{false} in general---already at $n=1$ it would require $D$ to preserve the range of $\widetilde V$---and is not used.]

\section{Small-sector values and strictness}\label{app:smalln}

\emph{Values at $n=0,1$.} Without dark counts, $F_C^{(0)}=0$ and $F_C^{(1)}=0$: no photon or one photon cannot produce clicks in two bases. This is consistent with Eq.~\eqref{eq:NCn}: at $n=1$, Eq.~\eqref{eq:AbRb} gives $N_C^{(1)}=\sum_bA_b-(m-1)A_\emptyset=A_\emptyset+\sum_bR_b=\id$. With dark counts, $1-f^{(0)}=\sum_b\gamma_b-(m-1)\gamma_\emptyset$, and $f^{(0)}>0$ precisely when detectors in at least \emph{two} different bases have nonzero dark-count probability---dark counts confined to a single basis cannot produce a cross click on the vacuum.

\emph{Exact $f^{(2)}$ for $M=2$.} By Eq.~\eqref{eq:NCn}, $N_C^{(2)}$ is the restriction of $\sum_b\gamma_bA_b^{\otimes2}-(m-1)\gamma_\emptyset A_\emptyset^{\otimes2}$ to the three-dimensional symmetric subspace of $\mathbb{C}^2\otimes\mathbb{C}^2$; $f^{(2)}=1-\eigmax(N_C^{(2)})$, where $\eigmax(N_C^{(2)})$ is the largest root of an explicit cubic and is available in closed (if unwieldy) form; the bound \eqref{eq:f2bound} is its simple two-line surrogate.

\emph{Strictness.} If the coverage condition \eqref{eq:coverage} holds, then $\lVert A_b\rVert<1$ for all $b$ and the upper bound in Eq.~\eqref{eq:twosided} decays exponentially, so $f^{(n)}<f^{(n')}$ for all sufficiently separated $n<n'$. Strict positivity of $f^{(2)}$ requires more than $\eta_d>0$ for arbitrary analyzers: if some input polarization is routed exclusively to the detectors of a single basis, a two-photon state in that polarization never cross-clicks and $f^{(2)}=0$, even though all efficiencies are positive. For \emph{splitter-type} analyzers---those whose first stage is a state-independent splitter directing the light with probabilities $p_b$ into arms that each implement a complete measurement on $\mathfrak h$, as in Fig.~\ref{fig:setup}---strict positivity is quantitative:
\begin{equation}\label{eq:f2strict}
  f^{(2)}\;\ge\;2\Big(\sum_{b<b'}p_b\,p_{b'}\Big)\,\eta_{\min}^2\;>\;0,
  \qquad \eta_{\min}:=\min_d\eta_d .
\end{equation}
Indeed, because the splitter is state independent, the arm occupation of the two photons is multinomial with probabilities $\{p_b\}$ regardless of the (possibly entangled) polarization state; conditioned on the photons occupying two distinct arms $b\neq b'$, each arm holds exactly one photon and the joint click operator satisfies $\big(\sum_{d\in\mathbb B_b}\eta_d\Pi_d\big)\otimes\big(\sum_{d\in\mathbb B_{b'}}\eta_d\Pi_d\big)\succeq\eta_{\min}^2\,\id\otimes\id$ since each factor sums over a complete basis; and simultaneous clicks in two arms constitute a cross click. For the balanced six-state analyzer, Eq.~\eqref{eq:f2strict} gives $f^{(2)}\ge\tfrac23\eta_{\min}^2$; numerically we find this bound to be \emph{tight} (attained with equality) for uniform efficiencies $\eta_d\equiv\eta$, consistent with Corollary~\ref{cor:ideal} at $\eta=1$. We do not pursue strict monotonicity at every step, which plays no role in applications.

\section{Proof of Corollary~\ref{cor:sixstate}}\label{app:sixstate}

For the balanced six-state analyzer, $v_d=\tfrac{1}{\sqrt3}|e_d\rangle$ and the two projectors of basis $b'$ satisfy $|e_{b'1}\rangle\langle e_{b'1}|+|e_{b'2}\rangle\langle e_{b'2}|=\id$ and $|e_{b'1}\rangle\langle e_{b'1}|-|e_{b'2}\rangle\langle e_{b'2}|=\sigma_{b'}$, where $\sigma_{\mathbb Z},\sigma_{\mathbb X},\sigma_{\mathbb Y}$ are the three Pauli operators. Hence
\begin{equation}
  R_{b'}=\frac{1}{3}\Big(\bar\eta_{b'}\,\id+\frac{\eta_{b'1}-\eta_{b'2}}{2}\,\sigma_{b'}\Big),
\end{equation}
and, for fixed $b$,
\begin{equation}
  \sum_{b'\neq b}R_{b'}
  =\frac{1}{3}\sum_{b'\neq b}\bar\eta_{b'}\,\id
  +\sum_{b'\neq b}\frac{\eta_{b'1}-\eta_{b'2}}{6}\,\sigma_{b'} .
\end{equation}
The Pauli axes appearing in the sum are orthogonal, so the traceless part has Bloch length $\frac16[\sum_{b'\neq b}(\eta_{b'1}-\eta_{b'2})^2]^{1/2}$, giving
\begin{equation}
  \eigmin\Big(\sum_{b'\neq b}R_{b'}\Big)
  =\frac{1}{3}\sum_{b'\neq b}\bar\eta_{b'}
  -\frac{1}{6}\sqrt{\sum_{b'\neq b}(\eta_{b'1}-\eta_{b'2})^2},
\end{equation}
and Eq.~\eqref{eq:sixrate} follows from Eq.~\eqref{eq:normAb}. For the configuration of Fig.~\ref{fig:fn}(b), $\vec\eta=(0.9,0.75,1.0,0.6,0.85,0.7)$, one finds $\lVert A_{\mathbb Z}\rVert=0.5462$, $\lVert A_{\mathbb X}\rVert=0.5020$, $\lVert A_{\mathbb Y}\rVert=0.5295$, in agreement with direct diagonalization.

\end{document}